\documentclass[11pt]{article}

\usepackage[margin=1in]{geometry}
\usepackage[ruled, vlined, linesnumbered]{algorithm2e}
\usepackage{amsmath}
\usepackage{amsfonts}
\usepackage{amsthm}
\usepackage{bbm}
\usepackage[colorlinks = true, citecolor = blue, urlcolor = blue]{hyperref}
\usepackage{enumitem}
\usepackage{todonotes}
\usepackage[square]{natbib}
\usepackage{cleveref}
\usepackage{subcaption}
\usepackage[toc,page]{appendix}

\newcommand{\sA}{\mathcal{A}}

\newcommand{\sC}{\mathcal{C}}

\newcommand{\sL}{\mathcal{L}}
\newcommand{\sP}{\mathcal{P}}
\newcommand{\sS}{\mathcal{S}}
\newcommand{\sT}{\mathcal{T}}
\newcommand{\sV}{\mathcal{V}}

\newcommand{\sbr}[1]{\left\{#1\right\}}

\newcommand{\indic}[1]{\mathbbm{1}_{#1}}
\newcommand{\R}{\mathbb{R}}

\newcommand{\norm}[1]{\left\|#1\right\|}
\newcommand{\abs}[1]{\left|#1\right|}
\newcommand{\set}[1]{\left\{#1\right\}}

\newcommand{\Normal}{\mathcal{N}}

\newcommand{\zeros}{\mathbf{0}}
\newcommand{\ones}{\mathbf{1}}

\DeclareMathOperator{\Unif}{Unif}
\DeclareMathOperator{\subjto}{subject\;to}
\DeclareMathOperator*{\minimize}{minimize\;}

\DeclareMathOperator*{\argmin}{argmin}

\DeclareMathOperator*{\rank}{rank}
\DeclareMathOperator*{\cspan}{span}

\DeclareMathOperator*{\cone}{Cone}
\DeclareMathOperator*{\conv}{Conv}
\DeclareMathOperator*{\hull}{Hull}
\DeclareMathOperator*{\nullspace}{Null}
\DeclareMathOperator*{\interior}{int}

\author{James Yang\\Stanford University
   \and Trevor Hastie\\Stanford University}

\title{A Fast Coordinate Descent Method for High-Dimensional Non-Negative Least Squares using a Unified Sparse Regression Framework}

\newtheorem{theorem}{Theorem}[section]
\newtheorem{corollary}[theorem]{Corollary}
\newtheorem{lemma}[theorem]{Lemma}
\newtheorem{definition}[theorem]{Definition}
\theoremstyle{remark}

\begin{document}
\maketitle

\abstract{
We develop theoretical results that establish a connection across
various regression methods such as the non-negative least squares,
bounded variable least squares, simplex constrained least squares, and lasso.
In particular, we show in general that a polyhedron constrained least squares problem
admits a \emph{locally unique sparse} solution in high dimensions.
We demonstrate the power of our result by concretely quantifying
the sparsity level for the aforementioned methods.
Furthermore, we propose a novel coordinate descent based solver for NNLS in high dimensions
using our theoretical result as motivation.
We show through simulated data and a real data example that
our solver achieves at least a 5x speed-up from the state-of-the-art solvers.
}

\section{Introduction}\label{sec:intro}

Sparse regression methods have grown increasingly popular recently
with the prevalance of high-dimensional data and large computing resource.
Notably, the lasso \citep{Tibshirani1996} is widely used 
for feature selection as it often sets many coefficients exactly to zero.
It is particularly useful in ultra high-dimensional settings such as in genomics
\citep{He2022,Qian2020,Li2020} where feature selection is especially meaningful for interpretability.
The non-negative least squares (NNLS) \citep{Lawson1995}
is also a reliable tool for many practitioners across diverse fields
as it often provides a physically meaningful and interpretable solution 
due to its empirical sparsity
\citep{li:spike,slawski:nnls,chou2022nonnegativesquaresoverparametrization}.
Some theoretical work exists around understanding this behavior
\citep{Bruckstein:2008,Wang:2009,Wang:2011}.
A natural extension of NNLS is the bounded variable least squares (BVLS) \citep{stark1995bounded},
which extends the non-negativity constraint to any box constraint.
Similar sparsity occurs in high dimensions where many coefficients lie at their respective boundaries.
The simplex constrained least squares (SLS) 
is at the heart of synthetic controls in causal inference \citep{abadie:synthetic}
where the potential outcome of no intervention for the treatment group is estimated
using a weighted average of the donor pool's outcomes.
\citet{abadie:synthetic} notes that we often observe sparsity in the weighted average,
which is desirable for interpretability among practitioners
especially when the donor pool size is large.
Indeed, we see in \citet{lei2024inferencesyntheticcontrolsrefined}[Table 2]
that many real data examples have more donors than the pre-intervention time periods,
which guarantees the existence of a sparse solution per our results in \Cref{sec:theoretical}.
Finally, the inequality simplex constrained least squares (ISLS) is commonly seen in portfolio optimization,
e.g., a long-only portfolio under a Markowitz model 
with a leverage constraint so that the asset weights $w$ 
satisfy $w \geq 0$ and $\sum_i w_i \leq C$ for some $C > 0$.
The model may be in a high-dimensional setting if the covariance matrix of the assets $\Sigma \in \R^{p \times p}$ 
admits a low-rank representation $\Sigma = FF^\top$ where $F \in \R^{p \times k}$ and $k < p$.
In this case, we often observe near zero weights for many assets
using a standard interior point solver.

While these methods may seem disparate at first glance, 
they share a deep connection through the more general problem of polyhedron constrained least squares.
Our first contribution is a theoretical result in \Cref{sec:theoretical} that establishes
the existence of a ``sparse'' solution to any polyhedron constrained least squares problem,
thereby unifying our understanding of sparse solutions to the aforementioned methods.
The notion of sparsity is extended to mean that the solution 
\emph{binds} many linear inequality constraints as equality constraints in the polyhedron
(see \Cref{def:polyhedron}).
Moreover, we quantify the number of such binding constraints
and show the \emph{local uniqueness property}, which states that
such sparse solutions are unique within the affine subspace defined by the binding constraints.
We apply our results to NNLS, BVLS, SLS, ISLS, lasso, and box constrained lasso to show that
the promised sparsity is at least $p-\rank(X) + O(1)$ where $X \in \R^{n \times p}$ is the feature matrix.
Lastly, we argue that the local uniqueness property helps explain
the success of many active-set optimizers for these problems.
Our second contribution is a computational result in \Cref{sec:algorithm} that 
uses the previous result as motivation to derive 
a coordinate descent based solver for BVLS and NNLS with screening and an active-set strategy.
We show in \Cref{sec:exp} that our solver 
performs competitively in terms of accuracy and speed in high-dimensional settings.
\section{Preliminaries and notations}\label{sec:notations}

Let $\R$, $\R^n$, and $\R^{m\times n}$ denote the space of real numbers,
real vectors of size $n$, and real matrices of shape $m\times n$, respectively.
Denote $\overline{\R} := \R \cup \set{\pm\infty}$ as the extended real line.
We denote $\zeros$ and $\ones$ to be the vector or matrix of zeros and ones, respectively,
depending on the context.
We define $[n] := \set{1,\ldots, n}$ to be the set of integers 
ranging from $1$ to $n$.
For any vector $v \in \R^n$ and a set of indices $S \subseteq [n]$,
let $v_S$ be the subset of $v$ along indices in $S$.
We denote $v_{-S} \equiv v_{S^c}$ as the subset of $v$
along indices not in $S$.
Similarly, for any matrix $A$, we denote $A_{S}$ as the rows of $A$ in $S$
and $A_{-S} \equiv A_{S^c}$.
For any two vectors $x,y \in \R^n$, we define $x \leq y$ to be the element-wise
less-than comparison operator and similarly for other comparison operators.
Let $x \odot y$ denote the Hadamard product of two vectors. 
The $\ell_p$-norm of a vector $x$ is given by $\norm{x}_p := (\sum_{i=1}^n \abs{x_i}^p)^{1/p}$.
For a matrix $X \in \R^{m \times n}$,
$\cspan(X)$ denotes the linear span of the columns of $X$. 
For a real-valued function $f : \R^n \to \R$,
we denote its gradient as $\nabla f$ and the $i$th partial derivative as $\partial_i f$.
Given a set $C \subseteq \R^n$, 
$\indic{C}(x)$ is defined to be $1$ whenever $x \in C$ and $0$ otherwise,
unless stated otherwise.
The positive and negative parts of a value $x \in \R$ 
are given by $x_+ := \max(x, 0)$ and $x_- := \max(-x, 0)$, respectively.
We denote $\Pi_{C}(\cdot)$ as the projection operator onto the set $C$ whenever defined.
For a matrix $X$, we may abuse the notation and denote $\Pi_{X}(\cdot) \equiv \Pi_{\cspan(X)}(\cdot)$.
\section{Theoretical results}\label{sec:theoretical}

\subsection{Background}\label{sec:theoretical:background}

We first begin with the definition of a polyhedron.
Note that some authors define a polyhedron more generally
and refer to \Cref{def:polyhedron} specifically as a \emph{convex polyhedron}.
Since we always work with convex polyhedra in this paper, we simplify the terminology.

\begin{definition}[Polyhedron]\label{def:polyhedron}
$\sP \subseteq \R^p$ is a \emph{polyhedron} if 
there exists $A \in \R^{m \times p}$ and $b \in \R^m$ for some $m$
such that $\sP \equiv \set{x \in \R^p : Ax \leq b}$.
We say that the polyhedron $\sP$ is \emph{represented by $A$ and $b$}.
An element $x\in \sP$ \emph{binds the $i$th constraint} if $a_i^\top x = b_i$.
Moreover, we say $x \in \sP$ \emph{binds $k$ constraints} if there exists $S \subseteq [m]$ of size $k$
such that $x$ binds the $i$th constraint for every $i \in S$.
We say a constraint $a_i$ is non-zero if $a_i \neq \zeros$.
\end{definition}

It is sometimes useful to identify a polyhedron
with a set of ``non-trivial'' inequality constraints
and a set of equality constraints that are always satisfied by its elements.
Concretely, it is useful to represent a polyhedron $\sP$ as
\begin{align}
    \sP
    =
    \set{x \in \R^p : A_{-S} x \leq b_{-S}, \, A_S x = b_S}
\end{align}
where $S := \set{i \in [m] : a_i^\top x = b_i, \, \forall x \in \sP}$
is the set of equality constraints.
\Cref{lem:polyhedron-decomposition} shows that 
$\set{x \in \R^p : A_{-S} x \leq b_{-S}}$ has a non-empty interior,
which is a way to express 
that this set of inequality constraints is ``non-trivial.''
Moreover, the affine space $\set{x \in \R^p : A_S x = b_S}$ is uniquely determined.
This motivates our definition of a dimension of a polyhedron.

\begin{definition}[Dimension of a polyhedron]\label{def:polyhedron-dimension}
Let $\sP \subseteq \R^p$ be a non-empty polyhedron represented by $A \in \R^{m\times p}$ and $b \in \R^{m}$. 
Let $S := \set{i \in [m] : a_i^\top x = b_i, \, \forall x \in \sP}$.
We call $S$ the \emph{maximal affine constraints of $\sP$}.
We define the \emph{dimension of a polyhedron $\sP$} 
denoted as $\dim(\sP)$ to be $p-\rank(A_{S})$.
\end{definition}

Intuitively, $p-\rank(A_S)$ is the degrees of freedom in representing
an element $x \in \sP$, and hence, the effective dimension of the polyhedron.
Since the affine space determined by $S$ is uniquely determined,
$\dim(\sP)$ is well-defined.

Next, we extend the notion of a \emph{hull} in \Cref{def:hull}.

\begin{definition}[Hull]\label{def:hull}
Let $Q \in \R^{n \times p}$ be any matrix and $\sP \subseteq \R^p$ be any set.
The \emph{$\sP$-hull} of $Q$ is defined to be 
\begin{align}
    \hull(Q, \sP) := \set{x \in \R^n : x = Qw, \, w \in \sP}
\end{align}
If $\sP$ is empty, we take $\hull(Q, \sP) \equiv \R^n$.
An element of $\hull(Q, \sP)$ is called a \emph{$\sP$-combination of $Q$}.
We say that $x \in \hull(Q, \sP)$ is a \emph{$\sP$-combination of $Q$ with $w \in \sP$}
if $x = Qw$.
\end{definition}

\Cref{def:hull} is a bona fide extension of commonly used hulls.
For example, a convex hull of a set of points $Q$ is $\hull(Q, \sP)$
where $\sP$ is the unit simplex; a conical hull of $Q$ is where $\sP$ is the positive orthant.
This extension will become useful as we will soon discuss the $\sP$-hull of $Q$
when $\sP$ is a polyhedron.
As a first result, we show in \Cref{lem:polyhedron-linear}
that if $\sP$ is a polyhedron, then $\hull(Q, \sP)$ is also a polyhedron.

\subsection{Main results}\label{sec:theoretical:main}

In this section, we discuss our main results for showing the 
existence of a sparse polyhedral regression solution in a high-dimensional setting
and the local uniqueness property of such solutions.
We first give a high-level description of our results. 
Suppose $\sP$ is a polyhedron represented by $A \in \R^{m \times p}$ and $b \in \R^m$.
We will show that there always exists a solution to
the polyhedral regression
\begin{align}
    \minimize_{\beta \in \sP} \frac{1}{2} \norm{y - X \beta}_2^2
    \label{eq:polyhedral-regression}
\end{align}
that binds many constraints of $A$.
Moreover, we will show the local uniqueness property
that such a solution is unique when restricted
to the affine subspace defined by the binding constraints.

We first restrict our attention to the simple case of a polyhedron 
$\sP$ with a non-empty interior.
Hence, the maximal affine constraints $S$ is empty.
Suppose that $x \in \hull(Q, \sP)$ so that $x = Qw$ for some $w \in \sP$.
Our goal is to show that we can find a (possibly different) $\tilde{w} \in \sP$
that binds at least a certain number of constraints of $A$
while reproducing $x$, i.e., $x = Q \tilde{w}$.
\Cref{lem:caratheodory-hull-full} formalizes this claim
and quantifies the number of binding constraints.
Moreover, it also shows that the binding constraints form a full row rank sub-matrix of $A$,
a fact that will become useful when we prove the local uniqueness property in \Cref{thm:caratheodory-hull-local-uniqueness}.

\begin{figure}[t]
    \centering  
    \includegraphics[width=0.8\linewidth]{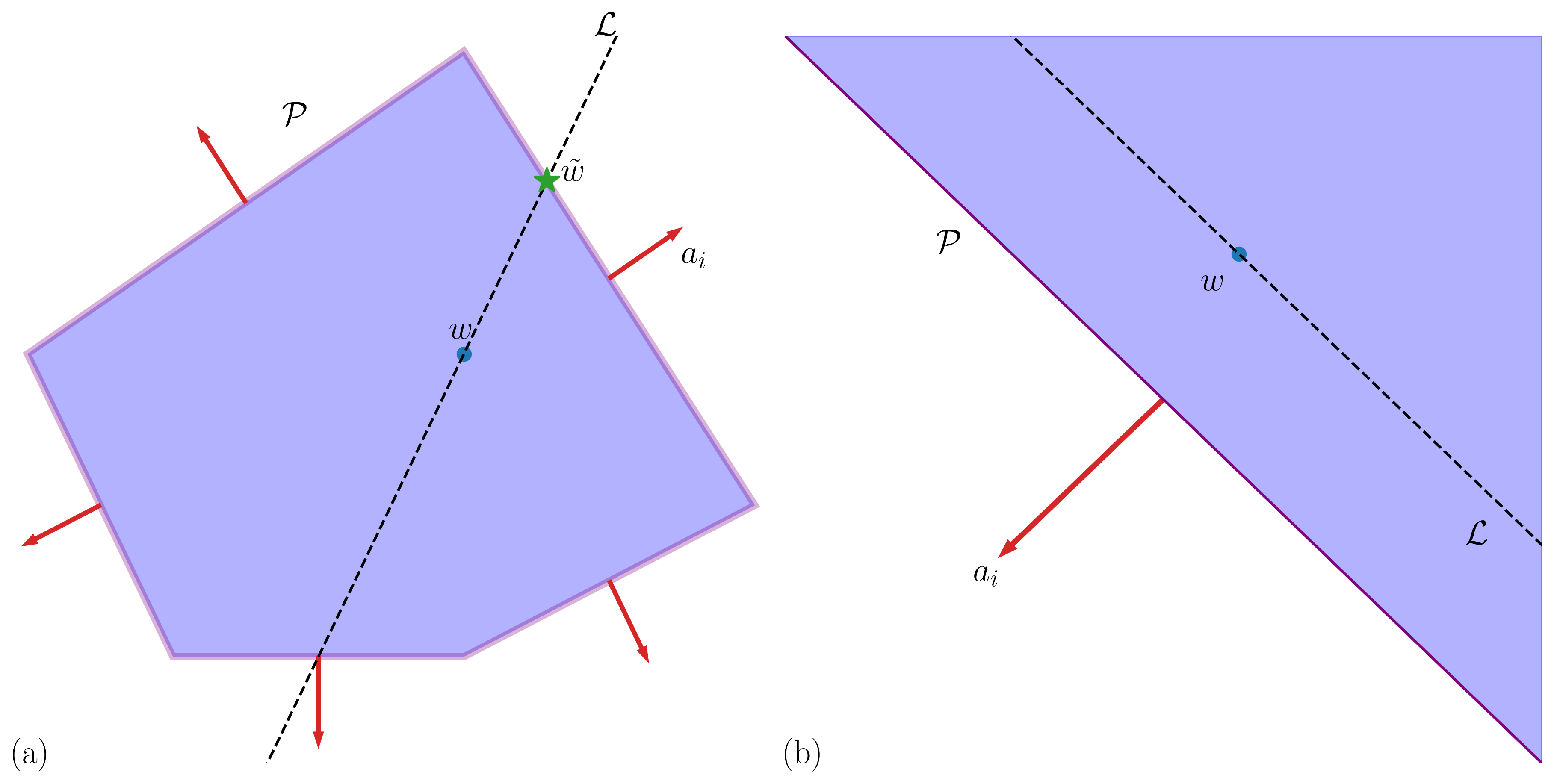}
    \caption{%
    Visualization of the main proof idea of \Cref{lem:caratheodory-hull-full}.
    The polyhedron $\sP \subseteq \R^2$ (with a non-empty interior) is given by the shaded blue region
    defined by the constraints $a_i$ (red arrows).
    Suppose $w \in \sP$ (blue dot) is such that $x = Qw$.
    If $v \neq \zeros \in \nullspace(Q)$, 
    then any point along the line $\sL := \set{w + tv : t \in \R}$
    (dotted black line) reproduces $x$.
    As long as the direction $v$ is not orthogonal to every constraint,
    i.e., $Av \neq \zeros$,
    then there always exists an intersection point $\tilde{w} \in \sL \cap \sP$
    by \Cref{lem:polyhedron-intersect}.
    The left plot shows the case of a polytope when $Av \neq \zeros$
    with the intersection point $\tilde{w}$ denoted by a green star.
    On the other hand, the right plot shows the opposite case 
    of a single half-space where no such an intersection point exists.
    }
    \label{fig:polyhedron-intersect}
\end{figure}

\begin{lemma}\label{lem:caratheodory-hull-full}
Let $Q \in \R^{n \times p}$ be any matrix
and $\sP \subseteq \R^p$ be any polyhedron represented by 
$A \in \R^{m \times p}$ and $b \in \R^m$ with a non-empty interior.
Suppose that $\nullspace(Q) \setminus \set{\zeros} \subseteq \nullspace(A)^c$.
If $x$ is in the $\sP$-hull of $Q$,
then $x$ can be written as a $\sP$-combination of $Q$ 
with an element $w$ that binds at least 
$\min(\rank(A), (p-n)_+)$ constraints of $A$
that form a full row rank sub-matrix of $A$.
\end{lemma}

We give a proof of \Cref{lem:caratheodory-hull-full} in
\Cref{app:proofs-main:caratheodory-hull-full}.
Note that \Cref{lem:caratheodory-hull-full} is only interesting 
in the high-dimensional setting of $p > n$.
The main idea of the proof is to consider a solution $w \in \sP$
such that $x = Qw$ and move $w$ along a direction in the null space of $Q$
until a boundary of $\sP$ is hit.
Once we hit a boundary, we may restrict ourselves to that facet of the polyhedron
and repeat the process in this lower-dimensional polyhedron.
The reader should imagine ``peeling'' off the facet,
discarding the rest of the polyhedron,
and viewing this facet in a lower-dimensional space.
Each time we restrict to the lower-dimensional space,
we remove a degree of freedom from both $Q$ and $\sP$.
Since $p > n$, we can then guarantee that at least 
$p-n$ degrees of freedom can be removed,
so long as there are enough degrees of freedom, i.e., $\rank(A)$, to remove from $\sP$.
With this intuition in mind, it is no surprise that \Cref{lem:caratheodory-hull-full} holds.

We now discuss a few technical aspects of \Cref{lem:caratheodory-hull-full}.
The hypothesis that $\nullspace(Q) \setminus \set{\zeros} \subseteq \nullspace(A)^c$
guarantees that moving along \emph{any} direction in $\nullspace(Q)$
will hit a boundary of $\sP$
with the help of \Cref{lem:polyhedron-intersect}.
This condition may seem convoluted at first, 
but it has a meaningful geometrical interpretation.
In words, any non-zero vector $v \in \nullspace(Q)$
must not be orthogonal to every constraint of $A$, i.e., $Av \neq \zeros$.
We give a visualization in \Cref{fig:polyhedron-intersect}.
We show in the left plot a case of a polytope where
any direction $v$ (not just in $\nullspace(Q)$) 
satisfies $Av \neq \zeros$.
On the other hand, the right plot depicts the case of a single half-space
where our condition may fail (specifically along the direction parallel to the boundary).
In the same vein,
the assumption that $\sP$ has a non-empty interior is also crucial,
since otherwise, the line $\sL = \set{w + tv : t \in \R}$ 
may only cross at $w$ 
while satisfying $Av \neq \zeros$
for every $v \in \nullspace(Q)$
(imagine $\sP$ on the xy-plane and a line piercing through $w$ not parallel to the z-axis).

We now remove the assumption that $\sP$ has a non-empty interior
while also strengthening the conclusion of \Cref{lem:caratheodory-hull-full}.
We show in \Cref{thm:caratheodory-hull} our main theoretical result,
which is an extension of \Cref{lem:caratheodory-hull-full} to the most general case.

\begin{theorem}[Existence of a sparse representation]\label{thm:caratheodory-hull}
Let $Q \in \R^{n \times p}$ be any matrix
and $\sP \subseteq \R^p$ be any polyhedron represented by $A \in \R^{m \times p}$ and $b \in \R^m$
with $S$ as its maximal affine constraints.
Further let $V \in \R^{p \times \dim(\sP)}$ be any full column rank matrix
such that $A_SV = \zeros$.
Suppose that $\nullspace(QV) \setminus \set{\zeros} \subseteq \nullspace(A_{-S}V)^c$.
If $x$ is in the $\sP$-hull of $Q$,
then $x$ can be written as a $\sP$-combination of $Q$ 
with an element $w$ that binds at least $\min(\rank(A)-\rank(A_S), \dim(\sP)-\rank(QV))$ 
constraints of $A_{-S}$ that form a full row rank sub-matrix.
\end{theorem}

We give a proof of \Cref{thm:caratheodory-hull} in
\Cref{app:proofs-main:caratheodory-hull}.
We first explain the assumptions of the theorem. 
As explained previously, a polyhedron may have an empty interior,
and therefore, \Cref{lem:caratheodory-hull-full} cannot be applied directly.
The remedy is the same as in the previous discussion;
we first find the subspace in which the polyhedron lies (i.e., given by $S$) and
remove the degrees of freedom from $Q$ and $\sP$ to restrict to a lower-dimensional space.
We then invoke \Cref{lem:caratheodory-hull-full} on the lower-dimensional space
where the polyhedron is guaranteed to have a non-empty interior by \Cref{lem:polyhedron-low-rank}.
The proof reveals that the new set of points is precisely $QV$ and 
the new polyhedron is represented by the constraint matrix $A_{-S}V$.
Moreover, we show in the proof that $\rank(A_{-S}V) = \rank(A)-\rank(A_S)$.
Lastly, we note that $\dim(\sP) - \rank(QV)$ has been improved from $\dim(\sP) - n$.
We argue that without loss of generality we may assume
$QV$ to be full row rank since otherwise we can perform a QR-decomposition
and use the ``$R$'' matrix in place of $QV$.

We provide a few consequences of \Cref{thm:caratheodory-hull}.
The first consequence is \Cref{cor:caratheodory}, which shows that the famous
Carath\'{e}odory's Theorem is an immediate corollary.
Secondly, we can finally establish the existence of a
``sparse'' solution for polyhedral regression.
This is formally stated in \Cref{cor:polyhedral-regression}.

\begin{corollary}[Carath\'{e}odory's Theorem]\label{cor:caratheodory}
Let $Q \in \R^{n \times p}$ be any matrix.
If $x$ is a convex combination of the columns of $Q$,
then it can be written as a convex combination of 
at most $n+1$ columns of $Q$.
Equivalently, if $x$ is a conical combination of the columns of $Q$,
then it can be written as a conical combination of
at most $n$ columns of $Q$.
\end{corollary}

\begin{proof}
Without loss of generality, we may assume $p > n$.

We first consider the case of convex combinations.
Define the polyhedron $\sP \subseteq \R^{p}$ represented by
\begin{align}
    A :=
    \begin{bmatrix}
        \ones^\top \\
        -I
    \end{bmatrix}
    ,\quad
    b :=
    \begin{bmatrix}
        1 \\
        \zeros
    \end{bmatrix}
\end{align}
Then, $S := \set{1}$ forms the maximal affine constraints.
It is clear that $\dim(\sP) = \rank(A) - \rank(A_S)= p-1$.
Then, for the matrix $V$ in \Cref{thm:caratheodory-hull},
which is injective, we have that
$\nullspace(QV) \setminus \set{\zeros} 
\subseteq \R^{p-1} \setminus \set{\zeros} = \nullspace(V)^c = \nullspace(A_{-S}V)^c$.
By definition, 
the convex hull of $Q$ is equivalent to the $\sP$-hull of $Q$.
Then, if $x$ is a convex combination of $Q$,
i.e., a $\sP$-combination of $Q$,
then by \Cref{thm:caratheodory-hull},
it is a $\sP$-combination of $Q$ with $w$
that binds at least $p-1-\rank(QV)$ constraints of $A_{-S}$.
In other words, $w \in \sP$ contains at most $\rank(QV) + 1 \leq n + 1$
positive values so that $x=Qw$ is a convex combination of at most $n+1$ columns of $Q$.

The second case of conical combinations is easier.
Define the polyhedron $\sP \subseteq \R^{p}$ represented by
\begin{align}
    A := -I ,\quad b := \zeros
\end{align}
It is clear that $\dim(\sP) = \rank(A) - \rank(A_S) = p$.
Since the maximal affine constraints is empty
and $A$ is invertible, the conditions of \Cref{thm:caratheodory-hull} are satisfied.
If $x$ is a conical combination of $Q$,
i.e., a $\sP$-combination of $Q$,
then by \Cref{thm:caratheodory-hull},
it is a $\sP$-combination of $Q$ with $w$
that binds at least $p-\rank(Q)$ constraints of $A$.
In other words, $w \in \sP$ contains at most $\rank(Q) \leq n$
positive values so that $x = Qw$ is a conical combination of at most $n$ columns of $Q$.
\end{proof}

\begin{corollary}[Polyhedral regression]\label{cor:polyhedral-regression}
Let $X \in \R^{n \times p}$ be a feature matrix 
and $y \in \R^{n}$ a response vector.
Suppose $\sP \subseteq \R^{p}$ is a polyhedron
represented by $A \in \R^{m\times p}$ and $b \in \R^{m}$
that satisfies the conditions of \Cref{thm:caratheodory-hull}
with $Q \equiv X$. 
Let $S$ be the maximal affine constraints of $\sP$
and $V \in \R^{p \times \dim(\sP)}$ such that $A_S V = \zeros$.
Consider the polyhedral regression given by \labelcref{eq:polyhedral-regression}.
Then, there exists a solution $\hat{\beta} \in \R^p$
that binds at least $\min(\rank(A)-\rank(A_S), \dim(\sP) - \rank(XV))$ constraints of $A_{-S}$
that form a full row rank sub-matrix.
\end{corollary}

\begin{proof}
Note that any solution $\beta$ for \labelcref{eq:polyhedral-regression}
must satisfy
\begin{align}
    \Pi_{\hull(X, \sP)}(y) = X \beta
\end{align}
where the projection exists since $\hull(X, \sP)$ is closed and convex,
which follows directly from \Cref{lem:polyhedron-linear}.
By \Cref{thm:caratheodory-hull},
since $\Pi_{\hull(X, \sP)}(y) \in \hull(X, \sP)$,
we have the desired claim.
\end{proof}

We conclude this section with a final result on 
\emph{the local uniqueness property} of polyhedrons in \Cref{thm:caratheodory-hull-local-uniqueness}.
The proof is provided in \Cref{app:proofs-main:caratheodory-hull-local-uniqueness}.

\begin{theorem}[Local uniqueness property]\label{thm:caratheodory-hull-local-uniqueness}
Consider the setting of \Cref{thm:caratheodory-hull}.
Suppose that 
\begin{align*}
    \rank(A) - \rank(A_S)
    \geq 
    \dim(\sP) - \rank(QV)
\end{align*}
Let $w$ be the sparse representation that binds at least $\dim(\sP) - \rank(QV)$ constraints of $A_{-S}$.
Let $T$ be the set of indices such that $(A_{-S})_T$ denotes the binding constraints.
Define $\sT := \set{w \in \R^p : (A_{-S})_T w = (b_{-S})_{T}}$.
Then, $w$ is the unique element in $\sP \cap \sT$ such that $x = Qw$.
\end{theorem}

The local uniqueness property says that
the sparse solution $w$ given by \Cref{thm:caratheodory-hull} 
must be unique in a sufficiently restricted subspace.
The condition for \Cref{thm:caratheodory-hull-local-uniqueness} can be equivalently written as
\begin{align}
    \rank(QV) + \rank(A) \geq p
\end{align}
This mild condition formalizes the idea that there must be enough degrees of freedom
in the set of points $QV$ and the constraint matrix $A$ combined
for the existence of a solution in such a sufficiently restricted subspace.
In many practical settings, this condition will hold true
as we often have $\rank(A) \geq p$.

The local uniqueness property offers a way of
concretely understanding active-set optimization algorithms.
Active-set strategies typicially employ heuristics
to first enforce many constraints and simplify an optimization problem
to a lower-dimensional problem.
If the binding constraints are correctly guessed,
then the solution to the lower-dimensional problem is optimal for the original problem.
These algorithms have found a place in the optimization literature
as sometimes the state-of-the-art methods for various polyhedral regression problems
such as the non-negative least squares and the lasso
\citep{friedman:2010,yang2024adelie,Lawson1995,MYRE2017755,stark1995bounded}.
\Cref{thm:caratheodory-hull-local-uniqueness} suggests that
if these strategies have an accurate way of 
guessing the binding constraints of an optimal solution, 
then the solution is unambiguous in the restricted subspace.
Hence, any optimizer should, in principle, quickly converge to the optimum.
In \Cref{sec:algorithm}, we will utilize this fact
to derive a new coordinate descent based solver for 
BVLS and NNLS that works especially well in high-dimensional settings.

\subsection{Applications on sparse regression models}\label{sec:theoretical:applications}

In this section, we use our main results from \Cref{sec:theoretical:main}
to unify many popular sparse regression models 
such as NNLS, BVLS, SLS, ISLS, lasso, and box constrained lasso
under the umbrella of polyhedral regression.
In all these examples, we establish the existence of a sparse solution.
In the case of lasso, we rediscover a well-known result
that there exists a solution with at most $\rank(X)$ non-zero entries \citep{tibshirani2012lassoproblemuniqueness}.
For all other models, we show an analogous result, 
which we believe is novel and does not exist in the literature
to the best of our knowledge.

\subsubsection{Non-negative least squares}

The non-negative least squares (NNLS) problem is given by
\begin{align}
    \minimize_{\beta \geq \zeros} \frac{1}{2} \norm{y-X\beta}_2^2
\end{align}
Using \Cref{cor:polyhedral-regression},
we now show in \Cref{cor:nnls-sparse} 
that there \emph{always} exists a solution with at most $\rank(X)$ positive entries.

\begin{corollary}[Existence of a sparse NNLS solution]\label{cor:nnls-sparse}
There exists an NNLS solution with at most $\rank(X)$ positive values.
\end{corollary}

\begin{proof}
The polyhedron $\sP$ is given by the positive orthant $\R_+^p$,
which can be represented by $A = -I$ and $b = \zeros$.
Since $A$ is invertible and the maximal affine constraint set is empty, 
the conditions of \Cref{thm:caratheodory-hull} are trivially satisfied.
Moreover, $\dim(\sP) = \rank(A) - \rank(A_S) = p$.
By \Cref{cor:polyhedral-regression},
there exists a solution $\hat{\beta} \in \R^p$
that binds at least $p-\rank(X)$ constraints of $A$.
That is, $\hat{\beta}$ has at most $\rank(X)$ positive values.
\end{proof}

\subsubsection{Bounded variable least squares}

The bounded variable least squares (BVLS) problem is given by
\begin{align}
    \minimize_{\ell \leq \beta \leq u}
    \frac{1}{2} \norm{y - X \beta}_2^2
    \label{eq:algorithm:bvls}
\end{align}
which can be thought of as an extension of NNLS.
Similar to NNLS, we show in \Cref{cor:bvls-sparse}
the existence of a sparse solution to BVLS.

\begin{corollary}[Existence of a sparse BVLS solution]\label{cor:bvls-sparse}
There exists a BVLS solution with at most $\rank(X)$ non-boundary values.    
\end{corollary}

\begin{proof}
The polyhedron $\sP$ is given by 
\begin{align}
    A := \begin{bmatrix}
        -I \\
        I
    \end{bmatrix}
    , \quad
    b := \begin{bmatrix}
        -\ell \\
        u
    \end{bmatrix}
\end{align}
Since $A$ has full column rank and the maximal affine constraint set is empty,
the conditions of \Cref{thm:caratheodory-hull} are trivially satisfied.
Moreover $\dim(\sP) = \rank(A) - \rank(A_S) = p$.
By \Cref{cor:polyhedral-regression}, there exists a solution
$\hat{\beta} \in \R^p$ that binds at least $p - \rank(X)$ constraints of $A$.
Note that a coefficient is either non-binding or exactly one of the lower or upper bounds is binding.
Hence, we indeed have that $\hat{\beta}$ has at most $\rank(X)$ non-binding entries.
\end{proof}

\subsubsection{Simplex constrained least squares}

The simplex constrained least squares (SLS) is given by
\begin{align}
    \minimize_{\beta \geq \zeros, \ones^\top \beta = C}
    \frac{1}{2} \norm{y - X \beta}_2^2
\end{align}
and the inequality version (ISLS) is given by
\begin{align}
    \minimize_{\beta \geq \zeros, \ones^\top \beta \leq C}
    \frac{1}{2} \norm{y - X \beta}_2^2
\end{align}
for some $C \geq 0$.
We show in \Cref{cor:sls-sparse} the existence of a sparse solution to both problems.

\begin{corollary}[Existence of a sparse ISLS/SLS solution]\label{cor:sls-sparse}
There exists an ISLS solution with at most $\rank(X)$ positive values,
and an SLS solution with at most $\rank(X)+1$ positive values. 
\end{corollary}

\begin{proof}
We first consider the ISLS problem.
The polyhedron $\sP$ is represented by
\begin{align}
    A := \begin{bmatrix}
        -I \\
        1^\top 
    \end{bmatrix}
    ,\quad 
    b := \begin{bmatrix}
        \zeros \\
        C
    \end{bmatrix}
\end{align}
Since $A$ has full column rank and the maximal affine constraint set is empty,
the conditions of \Cref{thm:caratheodory-hull} are trivially satisfied.
Moreover, $\dim(\sP) = \rank(A) - \rank(A_S) = p$.
By \Cref{cor:polyhedral-regression}, there exists a solution
$\hat{\beta} \in \R^p$ that binds at least $p-\rank(X)$ constraints of $A$.
If $\ones^\top \hat{\beta} < C$, then we have the desired claim.
Otherwise, there can be at most $\rank(X) + 1$ positive values in $\hat{\beta}$.
Without loss of generality, suppose $\hat{\beta}$ has exactly $\rank(X) + 1$ positive values.
Let $\sA := \set{j \in [p] : \hat{\beta}_j > 0}$.
Note that $X_{\cdot \sA}$ must have a non-trivial null space
since $\abs{\sA} = \rank(X) + 1$ by construction.
Consider any $v \neq \zeros \in \nullspace(X_{\cdot \sA})$.
Without loss of generality, we may assume $\ones^\top v \leq 0$.
Then, $v$ must have some strictly negative entry since otherwise $\ones^\top v \leq 0$
implies that $v \equiv \zeros$, which is a contradiction.
Then, $\ones^\top (\hat{\beta}_{\sA} + t v) \leq C$ for all $t \geq 0$
and there must exist a smallest $t$ (in magnitude) such that 
$\hat{\beta}_{\sA} + tv$ contains a zero while remaining in $\sP$.
Moreover, $X_{\sA} (\hat{\beta}_{\sA} + tv) = X_{\sA} \hat{\beta}_{\sA}$ by construction
so that the objective remains the same.
We have then constructed a possibly different optimal solution
with at most $\rank(X)$ positive values.

For the SLS problem,
the polyhedron $\sP$ is the same as in the ISLS case 
but the last constraint is always binding.
Since $A_{-S} \equiv -I$ is invertible,
the conditions of \Cref{thm:caratheodory-hull} are trivially satisfied.
Moreover, $\dim(\sP) = \rank(A) - \rank(A_S) = p-1$.
Let $V \in \R^{p \times \dim(\sP)}$ such that $A_S V = \zeros$.
By \Cref{cor:polyhedral-regression},
there exists a solution $\hat{\beta} \in \R^p$
that binds at least $p-1-\rank(XV) \geq p-1-\rank(X)$ constraints of $A_{-S}$.
That is, $\hat{\beta}$ has at most $\rank(X) + 1$ positive values. 
\end{proof}

\subsubsection{Lasso}

The lasso problem is given by
\begin{align}
    \minimize_{\beta} \frac{1}{2} \norm{y - X \beta}_2^2 + \lambda \norm{\beta}_1
    \label{eq:lasso}
\end{align}
for a regularization parameter $\lambda \geq 0$.
A perhaps interesting fact is that despite \labelcref{eq:lasso}
being an unconstrained problem, one can show that
there exists some $C \geq 0$ such that \labelcref{eq:lasso} is equivalent to the constrained problem
\citep{hastie01statisticallearning}
\begin{align}
    \minimize_{\norm{\beta}_1 \leq C} \frac{1}{2} \norm{y - X\beta}_2^2
    \label{eq:lasso-constrained}
\end{align}
When viewed this way, it becomes clearer why the lasso
exhibits very similar behavior to other polyhedral regression models such as NNLS.
Namely, the constraint $\norm{\beta}_1 \leq C$ defines a polytope.
In \Cref{cor:lasso-sparse},
we will use \labelcref{eq:lasso-constrained} to rediscover
the result that there exists a lasso solution with at most $\rank(X)$ non-zero values.
As a novel extension, we show in \Cref{cor:lasso-box-sparse}
that a similar result holds for the box constrained lasso problem given by
\begin{align}
    \minimize_{\ell \leq \beta \leq u} \frac{1}{2} \norm{y - X\beta}_2^2 + \lambda \norm{\beta}_1
    \label{eq:lasso-box}
\end{align}

\begin{corollary}[Existence of a sparse lasso solution]\label{cor:lasso-sparse}
There exists a lasso solution with at most $\rank(X)$ non-zero values. 
\end{corollary}

\begin{proof}
We use the constrained form \labelcref{eq:lasso-constrained}.
For any given lasso solution, it must lie in one of the $2^p$
disjoint partitions of the $\ell_1$-ball given by
\begin{align}
    \sP^j
    :=
    \set{\beta \in \R^p : A D^j \beta \leq b}
    , \quad
    A := \begin{bmatrix}
        -I \\
        \ones^\top
    \end{bmatrix}
    , \quad
    b := \begin{bmatrix}
        \zeros \\
        C
    \end{bmatrix}
\end{align}
where $D^j$ enumerates all diagonal sign matrices.
Hence, it suffices to show that the least squares 
constrained to any one of the $\sP^j$ polyhedrons
admits a solution with at most $\rank(X)$ non-zero values. 
Now, the least squares constrained to $\sP^j$ can be reparametrized as
\begin{align*}
    \minimize_{\beta} &\quad \frac{1}{2} \norm{y - X D^j \beta}_2^2 \\
    \subjto &\quad \beta \geq \zeros ,\quad \ones^\top \beta \leq C
\end{align*}
which is an ISLS problem.
By \Cref{cor:sls-sparse}, there exists a solution $\hat{\beta}$
with at most $\rank(XD^j) \equiv \rank(X)$ (since $D^j$ is invertible) positive values.
Then, the original solution is given by $D^j \hat{\beta}$,
which does not change the number of non-zero values.
\end{proof}

\begin{corollary}[Existence of a sparse box constrained lasso solution]\label{cor:lasso-box-sparse}
There exists a box constrained lasso solution 
with at most $\rank(X)$ non-zero, non-boundary values.
\end{corollary}

\begin{proof}
Consider the proof setup as in \Cref{cor:lasso-sparse}.
The only change is in $\sP^j$, which is now represented as
\begin{align}
    \sP^j 
    :=
    \set{\beta \in \R^p : AD^j \beta \leq b}
    ,\quad 
    A := \begin{bmatrix}
        -I \\
        I \\
        \ones^\top
    \end{bmatrix}
    ,\quad 
    b := \begin{bmatrix}
        \zeros \\
        \tilde{u}^j \\
        C
    \end{bmatrix}
\end{align}
where $\tilde{u}^j_k = -\ell_k$ if $D^j_{kk} = -1$
and otherwise $u_k$.
The $\sP^j$ constrained least squares problem can be reparametrized as
\begin{align*}
    \minimize_{\beta} &\quad \frac{1}{2} \norm{y - XD^j \beta}_2^2 \\
    \subjto &\quad \zeros \leq \beta \leq \tilde{u}^j, \quad \ones^\top \beta \leq C
\end{align*}
By \Cref{cor:polyhedral-regression},
there exists a solution $\hat{\beta}$ that binds at most
$p-\rank(XD^j) \equiv p - \rank(X)$ constraints.
Note that $\beta_k$ is either non-binding or binding the lower or upper bound, exclusively.
If $\ones^\top \beta < C$, then we have our desired claim.
Otherwise, we may use the same proof as in \Cref{cor:sls-sparse} for this case.
\end{proof}

\section{Coordinate descent for BVLS and NNLS}\label{sec:algorithm}

Consider the usual regression setup where $X \in \R^{n \times p}$
and $y \in \R^n$ are the feature matrix and response vector, respectively.
Let $\ell \leq u$ be vectors in $\overline{\R}^p$ 
denoting our lower and upper bounds (which may not be all finite) for the coefficients.
In this section, we discuss our algorithm for solving BVLS given by \labelcref{eq:algorithm:bvls}.
Note that BVLS is an extension of NNLS 
where $\ell \equiv \zeros$ and $u$ infinite. 
Our proposed algorithm is specifically designed to tackle
high-dimensional settings where $p > n$.
Although our only use-case of our method is in solving NNLS,
the algorithm extends naturally to the case of BVLS, so we present the more general algorithm.
While BVLS may be an ill-posed problem if $p > n$,
there nonetheless always exists a solution to \labelcref{eq:algorithm:bvls}
since the objective is closed, convex, and proper and the constraint is convex.

The most popular optimizer for NNLS is an active-set strategy 
first developed by Lawson and Hanson (LH-NNLS) \citep{Lawson1995}.
They propose an algorithm that searches for
the next variable to be added as a candidate non-zero variable.
Then, the unconstrained least squares solution is obtained for the set of non-zero variables.
A safety measure is put in place in case a variable is selected incorrectly.
Moreover, the algorithm is guaranteed to converge in finite number of steps.
Despite the simplicity of this algorithm,
it is now considered an out-dated algorithm for the scale of problems we see today.
Since then, many similar algorithms improve upon LH-NNLS
while adhering to the active-set strategy \citep{bro:fnnls,van2004fast}.
For example, 
the FNNLS method \citep{bro:fnnls} is a direct modification of LH-NNLS
designed to cache certain quantities for reuse during subsequent similar computations.
Other methods such as the TNT-NN method \citep{MYRE2017755}
further improve LH-NNLS by solving
the least squares solution using a left-preconditioned conjugate gradient descent normal residual routine.
We found one reference that uses coordinate descent to solve NNLS \citep{franc2005sequential};
however, it does not incorporate clever screening or active-set strategies,
which is a large part of our method.

Our algorithm for solving \labelcref{eq:algorithm:bvls} 
is based on the coordinate descent algorithm \citep{tseng:2001}.
Indeed, 
letting $\indic{\beta_i \in C}$ 
to be defined as $0$ whenever $x \in C$ and $+\infty$ otherwise,
we may rewrite \labelcref{eq:algorithm:bvls} as the following unconstrained problem:
\begin{align*}
    \minimize_{\beta}
    \frac{1}{2} \norm{y - X \beta}_2^2 
    + \sum\limits_{i=1}^p \indic{\ell_i \leq \beta_i \leq u_i}
\end{align*}
where the objective is composed of a differentiable convex loss
and a non-differentiable separable convex penalty.
In this setting, \citet{tseng:2001} provides a convergence guarantee.
We choose coordinate descent since it is the state-of-the-art method for fitting sparse regression models 
such as the lasso and group lasso \citep{yang2024adelie,friedman:2010}.
With the help of \Cref{cor:bvls-sparse},
we are assured that there must exist a sparse solution
with few coefficients in the interior of the box constraint.
Hence, we expect coordinate descent to efficiently find a sparse solution,
though no such guarantees exist.
Indeed, we observe in some cases, especially when $y$ is perfectly reproducible,
that the solution contains more than $\rank(X)$ non-binding entries.
However, in the cases when $y$ is not perfectly reproducible,
we often observe a sparse solution.

The vanilla coordinate descent algorithm 
simply cycles through each coefficient $\beta_k$
and performs a \emph{coordinate update} on $\beta_k$,
that is, optimize over the $k$th coordinate while holding all others fixed.
The separability and simplicity of the box constraints
allow us to implement an efficient closed-form coordinate update.
While it is possible in principle to solve BVLS by cycling over all $p$ variables,
it generally suffers from a long computation time when 
$p > n$ (underdetermined system) for two main reasons.
First, due to the non-uniqueness of the solution (in general)
and collinearity of $X$, 
the iterates tend to wander for a long period of time before converging to a solution
if the algorithm is not well warm-started.
Secondly, the sheer computational complexity increases as $p$ grows 
since there are more variables to cycle over. 
Before addressing these issues in full detail,
we give a high-level summary of our method.
First, we identify a well-chosen (and small) subset of features $S$
that most require adjustments to their coefficients.
In the process, we check the Karush-Kuhn-Tucker (KKT) conditions
to verify whether the current $\beta$ is optimal.
If the KKT check does not pass, we solve the BVLS problem only on $S$.
This procedure is repeated until the KKT conditions are satisfied or $S \equiv [p]$.
We summarize the pseudocode in \Cref{alg:bvls}
and explain the steps in detail below.
We note that this type of active-set strategy is ubiquitous not only in 
the BVLS and NNLS literature \citep{Lawson1995,stark1995bounded,MYRE2017755},
but also in lasso type regression literature \citep{friedman:2010,zou:2005,yang2024adelie}.

\begin{algorithm}[t]
    \caption{BVLS}\label{alg:bvls}
    \KwData{$X$, $y$, $\kappa$, $\epsilon_{\mathrm{kkt}}$}
    $\beta = \ell \odot \indic{\abs{\ell} < \abs{u}} + u \odot \indic{\abs{\ell} \geq \abs{u}}$\;
    $\sS = \emptyset$\;
    $\sA = \emptyset$\;
    \While{not converged and KKT check not passed} {
        Compute and sort violations $\delta$ as in \labelcref{eq:violations}\;
        $\abs{\sS_{\mathrm{old}}} = \abs{\sS}$\;
        \For{$j=1, \ldots, p$} {
            \If{$\delta_{(p-j+1)} \leq \epsilon_{\mathrm{kkt}} \norm{X_{\cdot,j}}_2 \norm{y}_2$ 
            or $(p-j+1) \in \sS$}{
                continue\;
            }
            Set a flag to mark that KKT check did not pass\;
            \If{$\abs{\sS} - \abs{\sS_{\mathrm{old}}} == \kappa$}{
                break\;
            }
            Add $(p-j+1)$ into $\sS$\;
        }
        \If{KKT check passed} {break\;}
        \While{not converged on $\sS$} {
            One cycle of coordinate descent on $\sS$
            and add all active variables to $\sA$\;
            \If{converged}{
                Remove binding coefficients from $\sA$\;
                break\;
            }
            Coordinate descent on $\sA$ until convergence\;
            Remove binding coefficients from $\sA$\;
        }
    }
    return $\beta$\;
\end{algorithm}

We first discuss how to choose the subset of features $S$.
Let $f(\beta) = \norm{y - X \beta}_2^2 / 2$ be the objective of BVLS.
Given an initial value $\beta$, we compute the gradient
$\nabla f(\beta) = -X^\top (y - X \beta)$.
The KKT conditions for optimality state that $\beta$ is optimal
if and only if 
\begin{align}
    \forall i \in [p]
    ,\quad 
    \partial_i f(\beta) 
    \in
    \begin{cases}
        [0,\infty) ,& \beta_i = \ell_i \\
        0 ,& \beta_i \in (\ell_i, u_i) \\
        (-\infty, 0] ,& \beta_i = u_i
    \end{cases}
    \label{eq:algorithm:bvls-kkt}
\end{align}
Hence, it is intuitive to first order the coefficients by their violation of the KKT conditions.
Namely, we compute the \emph{violations} 
\begin{align}
    \forall i \in [p]
    ,\quad 
    \delta_i 
    :=
    (\partial_i f(\beta))_- \indic{\beta_i < u_i}
    +
    (\partial_i f(\beta))_+ \indic{\beta_i > \ell_i}
    \label{eq:violations}
\end{align}
and sort in decreasing order from the most violating feature to the least.
We say that feature $i$ is \emph{violating} if $\delta_i > 0$.
While this heuristic provides a reasonable way to determine the order of features
to include in the set $S$, it is not clear thus far \emph{how many} to pick.
Fortunately, \Cref{cor:bvls-sparse} shows that 
there always exists a solution with $\rank(X)$ 
entries in the interior of their respective bounds.
A simple rule that works well in practice
is to batch at most $\rank(X)$ most violating features at a time.
If $\rank(X)$ is not known, one can approximate it with $\min(n, p)$.
In a sense, this batching creates a ``path'' of coefficients
similar to a penalized regression coefficient path (e.g., lasso) as the penalization is varied.

In the process,
we check the KKT conditions on $S^c$ to ensure that we have found an optimal solution.
We simply check that $\nabla f(\beta)$ satisfies \labelcref{eq:algorithm:bvls-kkt},
or equivalently, the violations are $0$.
For numerical stability, we relax the check with a tolerance that scales
with $X$ and $y$.
If the check passes, we have found an optimal solution.
Otherwise, we continue the algorithm with the augmented $S$.
Since $S$ is always strictly augmented,
we have guaranteed convergence via the coordinate descent algorithm
in the worst case that $S \equiv [p]$.
In practice, we rarely reach the worst case especially when $p \gg n$.

Given a subset of features $S$, we next solve BVLS only iterating over $S$.
We discuss a few optimizations to aid in faster convergence.
We keep track of the active set,
which contains the variables in $S$ that are in the interior of their bounds.
We first perform one cycle over $S$
and collect the active set in the process.
If we converged, we are done.
Otherwise, we cycle over only the active set until convergence.
We then remove variables in the active set that are at the boundary
and repeat the entire procedure.
This strategy significantly improves the performance of the algorithm
as $S$ tends to be sufficiently larger than the active set. 

We set the initial value $\beta$ 
with $\beta_i = \ell_i$ if $\abs{\ell_i} < \abs{u_i}$
and otherwise $u_i$.
This way, $\beta$ is always initialized at some vertex of the box
and is closest to zero entry-wise.
Since collinearity is inevitable in high-dimensional settings,
it is beneficial to initialize $\beta$ close to $\zeros$
to mimic the shrinkage in penalized regression.
\section{Numerical experiments}\label{sec:exp}

In this section, we discuss numerical experiments for our
algorithm in \Cref{sec:algorithm}.
We restrict our attention to NNLS for simplicity
and since it is more widely used in practice than the general BVLS problem.
We begin with a simulated data setting in \Cref{sec:exp:sim} and 
then move to a positive spike deconvolution example on real data in \Cref{sec:exp:spike}.
All experiments are run on a standard Macbook M1 Pro.
We implement an efficient version of our algorithm in C++
and include it in the Python package \texttt{adelie}\footnote{https://github.com/JamesYang007/adelie}.

\subsection{Simulated Data}\label{sec:exp:sim}

Our simulation setting is as follows.
We let $X \in \R_+^{n \times p}$ be a non-negative matrix
and $y \in \R^n$ be any vector.
This ensures that the conic hull of $X$ does not span $\R^n$
so that the NNLS problem is non-trivial.
In practice, there are also numerous examples such as 
text mining and functional MRI where the data is non-negative \citep{diakonikolas2022fast}.
We generate each $X_{ij} \sim \Unif(0, 1)$ independently
and set 80\% of the entries of $X$ to zero.
This results in a matrix with pairwise nearly orthogonal vectors.
We generate i.i.d. $y_i \sim \Normal(\mu, \sigma^2)$
where $\mu$ is chosen in the equally-spaced gridding of $[-3\sigma, 3\sigma]$ with $20$ points
and $\sigma = 1$.
We solve NNLS using our method (adelie) and MOSEK \citep{mosek}.
We note that both methods guarantee primal feasibility.
We also ran LH-NNLS and FNNLS \citep{Lawson1995,bro:fnnls},
but we omit their results since they took significantly longer.
\Cref{fig:simulation} shows a comparison
of the runtime and objective for when $n = 100$ and $n = 1000$.
In all cases, $p = 10000$.
In general, our method obtains a solution as accurate as MOSEK's
with a speed-up factor anywhere from 10-1000x.
As MOSEK is a general-purpose solver, 
it is expected that a specialized solver for NNLS will be faster;
however, we believe the scale of improvement is notable.
Moreover, it is clear that our solution tends to be quite sparse
while MOSEK always outputs a fully dense solution since 
their solver is based on interior point method.
However, as $y$ begins to lie in the cone of $X$,
i.e., as $\mu$ increases,
it is possible for our method to include more variables than necessary.
This is shown in the case of $n=100$ when $\mu$ is large
since there are more than $\min(n,p)$ active variables.

\begin{figure}[t]
    \centering 
    \includegraphics[width=\linewidth]{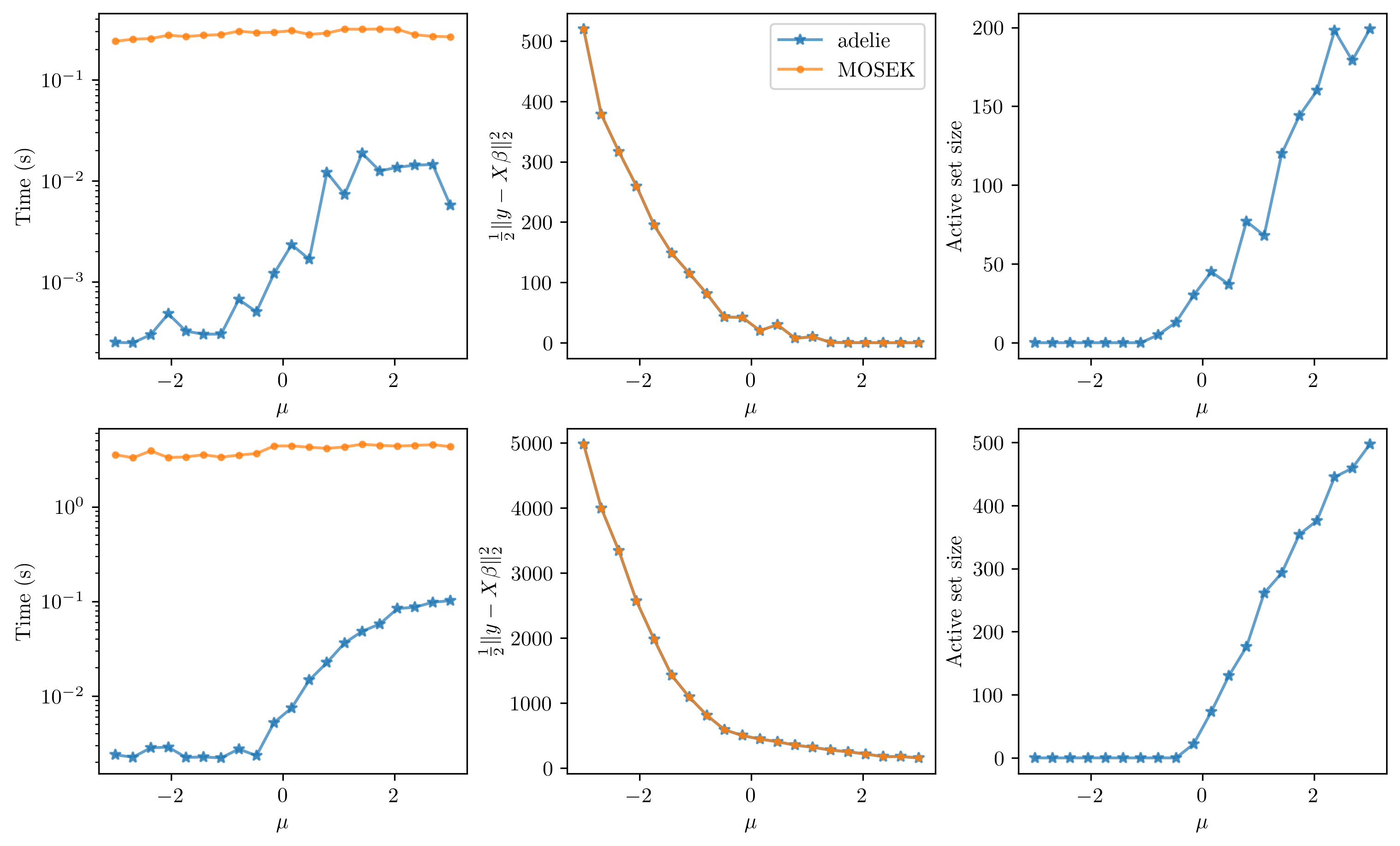}
    \caption{%
    Plots of solving NNLS in the simulation setting
    for various values of $\mu$ using our method (adelie) and MOSEK.
    The top three plots correspond to $n = 100$
    and the bottom three correspond to $n = 1000$.
    In all cases, $p = 10000$.
    The left column shows the time in seconds (smaller is better)
    and the middle column shows the objective
    for adelie and MOSEK.
    The right column shows the number of strictly positive entries 
    in the solution obtained via our method.
    In general, our method obtains a solution as accurate as MOSEK's
    with a speed-up factor anywhere from 10-1000x.
    Moreover, it is clear that our solution tends to be quite sparse
    while MOSEK always outputs a fully dense solution since 
    their solver is based on interior point method.
    However, as $y$ begins to lie in the cone of $X$,
    i.e., as $\mu$ increases,
    it is possible for our method to include more variables than necessary.
    This is shown in the case of $n=100$ when $\mu$ is large
    since there are more than $\min(n,p)$ active variables.
    }
    \label{fig:simulation}
\end{figure}

\subsection{Deconvolution of spike trains on carbon monoxide levels}\label{sec:exp:spike}

We consider a positive spike-deconvolution model given in \citet{li:spike}.
Suppose we observe a response vector $y \in \R^n$ where each $y_i$
are noisy observations of the mean function
\begin{align*}
    f(t; \beta, \tau) := \sum_{k=1}^p \beta_k \phi(t-\tau_k)
\end{align*}
The coefficients $\beta \in \R_+^p$ are assumed to be non-negative
and $\phi(\cdot)$ is the mean-zero Gaussian density with a user-specified standard deviation.
In practice, the user defines $\tau \in \R^p$ as a dense gridding
of an interval of interest, and the goal is to identify
$\beta \geq \zeros$ that best reconstructs $y$.
The positive entries of $\beta$ then correspond to the important spike points in $\tau$.
The feature matrix $X \in \R^{n \times p}$ is then given by
\begin{align*}
    X := \begin{bmatrix}
        \phi(t_1 - \tau_1) & \dots & \phi(t_1 - \tau_p) \\
        \vdots & \ddots & \vdots \\
        \phi(t_n - \tau_1) & \dots & \phi(t_1 - \tau_p) 
    \end{bmatrix}
\end{align*}
where $y_i = f(t_i; \beta, \tau) + \varepsilon_i$ for some noise $\varepsilon_i$.
We use NNLS to solve for $\beta$ by regressing $y$ on $X$ (given the non-negativity constraint). 

\begin{figure}[t]
    \centering 
    \includegraphics[width=\linewidth]{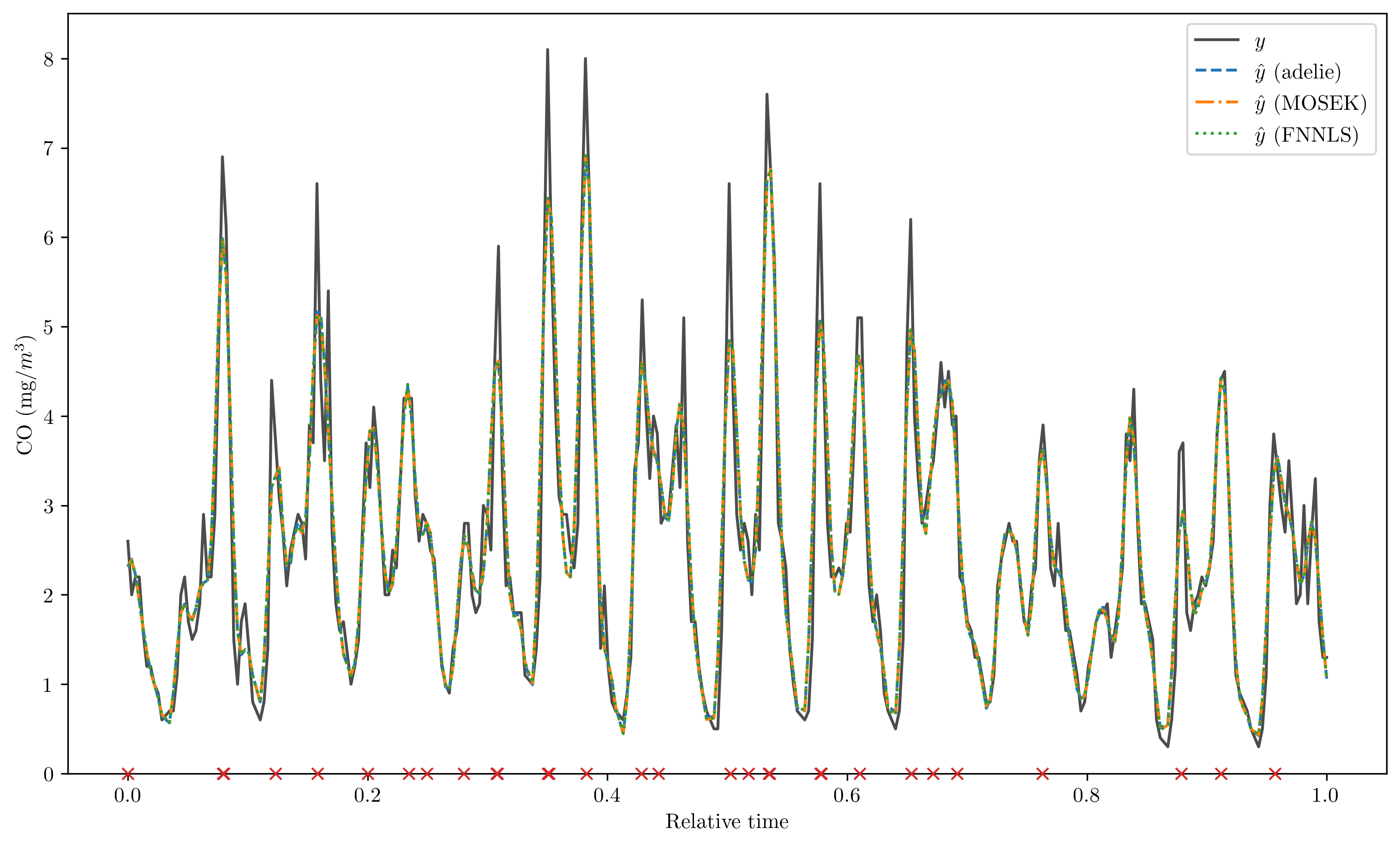}
    \caption{%
    Plot of the carbon monoxide level $y$ within the first two weeks of measurement
    and the NNLS predictions $\hat{y}$ using adelie, MOSEK, and FNNLS.
    We see that NNLS is able to reproduce $y$ fairly well and identify the spikes.
    The red ``x'''s mark the top 30 spike locations $\tau_k$ with the largest coefficients $\beta_k$
    using our method.
    These spike locations align well with the spikes observed in the data.
    All methods produce nearly identical fits.
    }
    \label{fig:spike-deconvolution}
\end{figure}

We use the Air Quality dataset from UCI \citep{air_quality_360} to apply the spike deconvolution method.
This dataset consists of hourly averaged responses from 5 metal oxide chemical sensors
in an Air Quality Checmical Multisensor Device located in a significantly polluted area of an Italian city.
We use only two weeks worth of data from March 10, 2004 to March 24, 2004.
Our goal is to detect the spikes of carbon monoxide levels (mg/$m^3$).
After removing missing data, there are $n=304$ samples.
The time stamps are put relative to the start date and normalized to the unit interval.
We let $p = 1000$ and set $\tau \in \R^p$ equally spaced on the interval $[0, 1]$.
The standard deviation in $\phi$ is defined as 
2 times the median (normalized) time difference in the $n$ samples,
as similarly done in \citet{slawski:nnls}.
We solve the NNLS problem using our method, which takes 0.026 seconds.
For a comparison, we also solve it using MOSEK and FNNLS \citep{bro:fnnls}, 
which take 0.12 seconds and 1.17 seconds, respectively.
In \Cref{fig:spike-deconvolution}, we show the response $y$ along with the predictions $X \hat{\beta}$
where $\hat{\beta}$ is obtained via our method, MOSEK, and FNNLS.
We see that the fits are nearly identical and reproduce $y$ fairly well.
Although all methods produce a feasible solution, 
our method and FNNLS are the only ones that output sparse solutions.
FNNLS is the most accurate method, but it is the most time consuming.
MOSEK is only slightly less accurate despite being about 10x faster, 
however, it outputs a dense solution, which is not as easily interpretable.
We see our method as a bridge between the two methods;
that is, we provide a sparse solution with a similar accuracy as MOSEK and at a fast speed.
Using our method, we identify $225$ spikes for this example.
\section{Discussion}\label{sec:discussion}

We develop a framework for understanding sparse solutions to any polyhedron constrained least squares problem.
In particular, we establish a connection between popular regression methods
such as NNLS, BVLS, SLS, ISLS, and the lasso that they all admit a locally unique sparse solution in high dimensions.
We leverage this theory to derive an efficient coordinate descent based solver for NNLS,
demonstrating that our solver tends to quickly find such sparse solutions
in both simulated data settings and real data examples.

Our implementation is included in the Python package \texttt{adelie},
which is available under the MIT license from PyPI
(\href{https://pypi.org/project/adelie}{https://pypi.org/project/adelie}).
The user can install the latest version by running
\texttt{pip install adelie} in the terminal.

\section*{Acknowledgments}

Trevor Hastie was partially supported by grants DMS-2013736 from the National Science Foundation, 
and grant R01GM134483 from the National Institutes of Health. 
We thank Jonathan Taylor for comments that helped solidify a rigorous grounding of our theoretical work.
Lastly, we also thank 
Tim Morrison, John Cherian, Kevin Guo, and Anav Sood for insightful discussions during the development of this paper.

\bibliography{references}

\newpage 

\begin{appendix}

\section{Background}\label{app:background}

\begin{lemma}[Interior of a polyhedron]\label{lem:polyhedron-interior}
Let $\sP \subseteq \R^p$ be a polyhedron represented by $A$ and $b$. 
Then, the interior of $\sP$ is given by $\interior(\sP) = \set{x \in \sP : Ax < b}$.
\end{lemma}

\begin{proof}
Without loss of generality, assume $\sP$ is non-empty.
Suppose $x \in \interior(\sP)$.
Then, there exists $B_r(x)$, a ball of radius $r$ centered at $x$, for some $r > 0$
such that $B_r(x) \subseteq \interior(\sP)$.
Suppose $x$ binds a constraint $a_i$ for some $i \in [m]$.
Without loss of generality, we may assume $a_i \neq \zeros$
since it may be removed from the constraints without changing $\sP$.
Consider $w := x + \epsilon a_i$ for any $\epsilon > 0$.
Then, $a_i^\top w = a_i^\top x + \epsilon \norm{a_i}_2^2 = b_i + \epsilon \norm{a_i}_2^2 > b_i$
so that $w \notin \sP$.
In particular, if $\epsilon < r / \norm{a_i}_2$,
then $w \in B_r(x)$, which is a contradiction.

Conversely, suppose $x \in \R^p$ such that $Ax < b$.
Consider $\varphi$ defined by
\begin{align}
    \varphi(t) 
    &:= 
    \max_{i \in [m]} \max_{v : \norm{v}_2=1}\sbr{a_i^\top (x + tv) - b_i}
    \\&=
    \max_{i \in [m]} \sbr{a_i^\top x + t \norm{a_i}_2 - b_i}
\end{align}
which identifies the constraint closest to the boundary
on the ball centered at $x$ with radius $t$.
Note that $\varphi$ is continuous, $\varphi(0) < 0$, and $\varphi(t) \to \infty$ as $t\to \infty$.
Hence, there exists $t^\star > 0$ such that $\varphi(t^\star) = 0$.
Note also that $\varphi$ is strictly increasing, 
so this implies that for any $t \in [0, t^\star)$,
$\varphi(t) < 0$.
In other words, $B_t(x) \subseteq \sP$ so that $x$ is an interior point.
\end{proof}

\begin{lemma}\label{lem:polyhedron-decomposition}
Let $\sP \subseteq \R^p$ be a non-empty polyhedron represented by $A \in \R^{m \times p}$ and $b \in \R^m$.
Define $S := \set{i \in [m] : a_i^\top x = b_i ,\, \forall x \in \sP}$
as the set of constraints that are always binding.
Then, the polyhedron represented by $A_{-S}$ and $b_{-S}$ has a non-empty interior
and the affine space $\{x \in \R^p : A_S x = b_S\}$ is unique.
Consequently, any non-empty polyhedron can be decomposed as
an intersection of a polyhedron with a non-empty interior and a unique affine space.
\end{lemma}

\begin{proof}
We first show that the polyhedron represented by $A_{-S}$ and $b_{-S}$ has a non-empty interior.
If $\sP$ has a non-empty interior, we are done. 
Suppose $\sP$ has an empty interior.
We claim that there exists a constraint $a_{i^\star}$ such that
$a_{i^\star}^\top x = b_{i^\star}$ for all $x \in \sP$.
Suppose not so that for every constraint $a_i$, there exists $x_i \in \sP$ such that $a_i^\top x_i < b_i$.
We show that there exists a convex combination of $\set{x_k}_{k=1}^m$ such that $Ax < b$, 
which will establish a contradiction by \Cref{lem:polyhedron-interior}.
Indeed, for any $\theta \in \R_{++}^m$ strictly positive and $\sum_k \theta_k = 1$,
\begin{align}
    a_i^\top \sum_k \theta_k x_k
    &<
    b_i
    ,\quad \forall i \in [m]
\end{align}
since for every $i \in [m]$, $a_i^\top x_i < b_i$
while $a_i^\top x_k \leq b_i$ for all $k \neq i$.
This completes the proof for showing there exists a constraint
$a_{i^\star}$ such that $a_{i^\star}^\top x = b_{i^\star}$ for all $x\in \sP$.
Once we identify the constraint $a_{i^\star}$, 
we remove it and repeat the procedure on the polyhedron
$\sP' := \set{x \in \R^p : A_{-i^\star} x \leq b_{-i^\star}}$
until $\sP'$ has a non-empty interior.
Clearly, this process must terminate in finite steps.

We now show that the affine space determined by $S$ is unique.
Suppose there exists a different affine space that contains $\sP$.
Then, $\sP$ must lie in a strictly lower-dimensional subspace,
which implies that $\sP$ has an empty interior.
This contradicts the fact that we have terminated our procedure
when we reached $\sP'$ with a non-empty interior.
\end{proof}

\begin{lemma}\label{lem:polyhedron-linear}
Let $Q \in \R^{n \times p}$ be any matrix and $\sP \subseteq \R^p$ be a polyhedron.
Then, $\hull(Q, \sP)$ is a polyhedron.
\end{lemma}

\begin{proof}
Without loss of generality, assume $\sP$ is non-empty.
By the polyhedron decomposition theorem,
$\sP$ can be represented as a Minkowski sum of a cone and a polytope.
Let $\sC = \cone(B_{\sC})$ and $\sV = \conv(B_{\sV})$ denote the cone and the polytope
generated by $B_{\sC}$ and $B_{\sV}$, respectively.
Then, $Q(\sP) = \cone(QB_{\sC}) + \conv(QB_{\sV})$.
By the polyhedron decomposition theorem again,
$Q(\sP)$ is a polyhedron.
\end{proof}

\begin{lemma}\label{lem:polyhedron-intersect}
Let $\sP \subseteq \R^p$ be a polyhedron represented by 
$A \in \R^{m \times p}$ and $b \in \R^m$ with a non-empty interior.
Suppose $y$ is any element in the interior of $\sP$ and $v \in \R^p$ any vector.
Then, $v \notin \nullspace(A)$ if and only if
there exists a point $w = y + tv$ for some $t \neq 0$ such that
$w$ lies in $\sP$ and binds at least one constraint of $A$.
\end{lemma}

\begin{proof}
Since $y$ is in the interior of $\sP$,    
we must have that $Ay < b$ by \Cref{lem:polyhedron-interior}.

Suppose $v \notin \nullspace(A)$ so that $Av \neq \zeros$.
Then, there exists $i \in [m]$ such that $a_i^\top v \neq 0$.
For any constraint $a_i$ not orthogonal to $v$,
\begin{align}
    a_i^\top (y + tv) = b_i
    \iff
    t = \frac{b_i - a_i^\top y}{a_i^\top v}
    \label{lem:polyhedron-intersect:eq:t}
\end{align}
Label the solution for $t$ in \labelcref{lem:polyhedron-intersect:eq:t} 
as $t_i^\star$ for each $i$ such that $a_i^\top v \neq 0$.
Then, we simply choose the $t_i^\star$ that is smallest in magnitude,
that is,
\begin{align}
    i^\star := \argmin_{i : a_i^\top v \neq 0} \abs{t_i^\star}
    ,\quad
    t^\star := t_{i^\star}^\star
\end{align}
Note that since each $t_{i}^\star \neq 0$,
we also have that $t^\star \neq 0$.
With this choice of $t^\star$,
we have the guarantee that $w := y + t^\star v$
binds the $i^\star$th constraint
and is the first binding constraint so that $w$ remains in $\sP$.

Conversely, suppose $v \in \nullspace(A)$ so that $Av = \zeros$.
Then, for any point $w = y + tv$ for any $t \in \R$,
we have that $Aw = Ay < b$.
Hence, every such point $w$ does not bind any constraint.
\end{proof}

\begin{lemma}\label{lem:polyhedron-low-rank}
Let $\sP \subseteq \R^p$ be a non-empty polyhedron 
represented by $A \in \R^{m \times p}$ and $b \in \R^m$
with $S$ as its maximal affine constraints.
Let $\xi \in \R^p$ be any vector satisfying $A_S \xi = b_S$
and $V \in \R^{p \times \dim(\sP)}$ be any full column rank matrix
such that $A_S V \equiv \zeros$.
Define the following quantities
\begin{align}
    A' := A_{-S} V \in \R^{(m-\abs{S}) \times \dim(\sP)}
    ,\quad
    b' := b_{-S} - A_{-S}\xi \in \R^{m-\abs{S}}
    \label{lem:polyhedron-low-rank:eq:low-rank}
\end{align}
and let $\sP' \subseteq \R^{\dim(\sP)}$ be the polyhedron 
represented by \labelcref{lem:polyhedron-low-rank:eq:low-rank}.
Then, every $x \in \sP$ can be represented uniquely as $x = \xi + Vx'$ with $x' \in \sP'$.
Moreover, $\sP'$ has a non-empty interior.
\end{lemma}

\begin{proof}
By hypothesis, $\sP \equiv \set{x \in \R^p : A_{-S} x \leq b_{-S} , \, A_S x = b_S}$.
First, it is a standard linear algebra fact that 
any $x \in \R^p$ such that $A_S x = b_S$ can be 
uniquely identified with the low-rank representation 
$x = \xi + Vx'$ for some $x' \in \R^{\dim(\sP)}$.
Then, we have the first claim by simple algebra that
\begin{align}
    A_{-S}x - b_{-S}
    =
    A'x' - b'
    \label{lem:caratheodory-hull-full:eq:id}
\end{align}

We now show that $\sP'$ has a non-empty interior.
Since $\set{x \in \R^p : A_{-S}x \leq b_{-S}}$ has a non-empty interior by \Cref{lem:polyhedron-decomposition}, 
there exists $x \in \sP$ such that $A_{-S} x < b_{-S}$
by \Cref{lem:polyhedron-interior}.
By \labelcref{lem:caratheodory-hull-full:eq:id}, we must have that $A'x' < b'$.
Another application of \Cref{lem:polyhedron-interior} shows that $\interior(\sP')$ is non-empty.
\end{proof}

\section{Proofs of the main results}\label{app:proofs-main}

\subsection{Proof of Lemma \ref{lem:caratheodory-hull-full}}\label{app:proofs-main:caratheodory-hull-full}

\begin{proof}
Without loss of generality, assume every constraint is non-zero
since otherwise we may remove these constraints without changing $\sP$ or $\nullspace(A)$.
If $p \leq n$, the result is trivial. 
Hence, we assume without loss of generality that $p > n$.

We first prove the theorem when $p = n+1$.
Suppose $x$ is in the $\sP$-hull of $Q$ so that
$x = Qy$ for some $y \in \sP$.
If $y$ binds at least one constraint of $A$, we are done.
Suppose then $y$ binds no constraint of $A$ so that it lies in the interior of $\sP$
by \Cref{lem:polyhedron-interior}.
Since there are more columns than rows in $Q$,
we necessarily have that $\nullspace(Q)$ is non-trivial.
Hence, there exists $v \neq \zeros$ such that $Qv = \zeros$.
Since $v \in \nullspace(Q) \setminus \set{\zeros} \subseteq \nullspace(A)^c$,
we have by \Cref{lem:polyhedron-intersect} that
there exists $w = y + tv$ for some $t \neq 0$ such that 
$w$ lies in $\sP$ and binds at least one constraint of $A$.
Finally, $Qw = x$ by construction.
This proves the theorem when $p = n+1$.

Now, suppose the theorem holds for $p > n$.
We now prove it holds for $p+1$.
Suppose $x$ is a $\sP$-combination of $Q = \set{q_i}_{i=1}^{p+1}$.
Consider the embedding 
\begin{align}
    \forall i \in [p+1],
    \quad
    u_i := \begin{bmatrix}
        q_i \\ 
        \zeros
    \end{bmatrix}
    \in \R^{p}
    ,\quad
    v := \begin{bmatrix}
        x \\
        \zeros
    \end{bmatrix}
    \in \R^{p}
\end{align}
where we pad our original points with zeros until the length is $p$.
Define $U := \set{u_i}_{i=1}^{p+1}$.
Then, $v$ is a $\sP$-combination of $U$.
By the preceding argument, noting that $\nullspace(Q) \equiv \nullspace(U)$,
$v$ can then be written as a $\sP$-combination
of $U$ with an element $w$ that binds at least one constraint of $A$.
Finally, we conclude that $x$ can be written as a $\sP$-combination of $Q$ 
with (the same) $w$ that binds at least one constraint of $A$.
Without loss of generality, suppose the binding constraint is the first constraint $a_1$.
Consider the polyhedron $\set{z \in \R^{p+1} : A_{-1} z \leq b_{-1}, \, a_1^\top z = b_1}$
whose dimension is then $p$ (since $a_1 \neq \zeros$).
By \Cref{lem:polyhedron-low-rank}, it can be represented 
with the polyhedron $\sP'$
defined by $A'$ and $b'$ in \labelcref{lem:polyhedron-low-rank:eq:low-rank}
with a non-empty interior.
Let $V$ and $\xi$ be as in \Cref{lem:polyhedron-low-rank}
and define $Q' := QV$ and $x' := x - Q\xi$.
Note that $x' \in \hull(Q', \sP')$.
We will apply the induction hypothesis on $Q'$, $\sP'$, and $x'$.
Indeed, $\nullspace(Q') \setminus \set{\zeros} \subseteq \nullspace(A')^c$
since for any $y \in \nullspace(Q') \setminus \set{\zeros}$,
$Vy \in \nullspace(Q)$, which implies that $AVy \neq \zeros$ by hypothesis.
Since $a_1^\top V = \zeros^\top$ by definition of $V$,
we have that $A'y = A_{-1}Vy \neq \zeros$ so $y \in \nullspace(A')^c$.
Hence, the assumptions of the theorem are satisfied.
By the induction hypothesis,
$x'$ is a $\sP'$-combination of $Q'$ with an element $w'$ 
that binds at least $\min(\rank(A'), p - n)$ constraints of $A'$
that forms a full row rank sub-matrix of $A'$.
Letting $w := \xi + Vw'$
and noting that
\begin{align}
    A_{-1} w - b_{-1} = A'w' - b',
\end{align}
$x$ is a $\sP$-combination of $Q$ with $w$
that binds at least $\min(\rank(A'), p-n)$ constraints of $A_{-1}$,
which we denote as $T \subseteq [m-1]$.
Hence, including the first binding constraint,
$w$ binds at least $\min(\rank(A')+1, p+1-n)$ constraints of $A$.
It remains to show that $\rank(A) \equiv \rank(A') + 1$
and that appending $a_1$ as a row to $(A_{-1})_T$ keeps it full row rank.
To show the first claim, observe that
\begin{align}
    A \begin{bmatrix}
        a_1 & V 
    \end{bmatrix}
    =
    \begin{bmatrix}
        Aa_1& AV
    \end{bmatrix}
    \label{lem:caratheodory-hull-full:eq:rankA}
\end{align}
By definition of $V$, $\begin{bmatrix}
    a_1 & V
\end{bmatrix}$ is invertible so that the rank of the left side of 
\labelcref{lem:caratheodory-hull-full:eq:rankA} is $\rank(A)$.
Note that $Aa_1 \notin \cspan(AV)$ since the first row of $AV$ 
is $a_1^\top V \equiv \zeros^\top$
while that of $Aa_1$ is $\norm{a_1}_2^2 \neq 0$.
Again using that $a_1^\top V = \zeros^\top$,
$\rank(AV) = \rank(A_{-1}V) = \rank(A')$.
This shows that $\rank(A) = \rank(A') + 1$.
Next, we show the second claim.
It suffices to show that $a_1$ is not in the row space of $(A_{-1})_T$.
Otherwise, there exists $\gamma \neq \zeros \in \R^{\abs{T}}$ 
such that $a_1 = (A_{-1})_T^\top \gamma$.
Then, 
\begin{align}
    \zeros = V^\top a_1 = V^\top (A_{-1})_T^\top \gamma 
    = ((A_{-1})_T V)^\top \gamma
    = (A_T')^\top \gamma
    \implies 
    \gamma = \zeros
\end{align}
which is a contradiction.
This shows that including $a_1$ as another row in $(A_{-1})_T$
keeps it full row rank.
By induction, this concludes the proof.

\end{proof}

\subsection{Proof of Theorem \ref{thm:caratheodory-hull}}\label{app:proofs-main:caratheodory-hull}

\begin{proof}

First, consider the setting of \Cref{lem:caratheodory-hull-full}.
We argue that the conclusion can be improved from 
$\min(\rank(A), (p-n)_+)$ to 
$\min(\rank(A), (p-\rank(Q))_+)$.
We first perform a QR-decomposition on $Q$ 
to get $Q = \tilde{Q} \tilde{R}$ 
with $\tilde{Q} \in \R^{n \times r}$, $\tilde{R} \in \R^{r \times p}$,
and $r := \rank(Q)$.
Note that $\tilde{Q}$ has orthogonal columns 
so that $\nullspace(Q) \equiv \nullspace(\tilde{R})$.
This shows that $\tilde{R}$ satisfies the assumptions of \Cref{lem:caratheodory-hull-full}.
Now, if $x \in \hull(Q, \sP)$,
there exists $w \in \sP$ such that $x = Qw = \tilde{Q} \tilde{R} w$.
We then apply \Cref{lem:caratheodory-hull-full} for $\tilde{R}$, $\sP$, and $\tilde{x} := \tilde{R} w$
to get that $\tilde{x}$ is a $\sP$-combination of $\tilde{R}$
with an element $\tilde{w} \in \sP$ that binds at least 
$\min(\rank(A), (p-r)_+)$ constraints of $A$
that form a full row rank sub-matrix.
Since $\tilde{R}w = \tilde{R}\tilde{w}$,
this immediately shows that $x = \tilde{Q} \tilde{R} w = \tilde{Q} \tilde{R} \tilde{w} = Q\tilde{w}$.
Hence, we have our improved result.

If $\sP$ has a non-empty interior,
then \Cref{lem:caratheodory-hull-full} gives us the desired conclusion
noting that $\dim(\sP) \equiv p$.
Suppose $\sP$ has an empty interior.
Then, $\rank(A_S) \neq 0$ 
so that \Cref{lem:polyhedron-low-rank} reduces our representation of $\sP$ 
to a lower-dimensional polyhedron $\sP'$ represented with $A'$ and $b'$
given in \labelcref{lem:polyhedron-low-rank:eq:low-rank}.
Moreover $\sP'$ has a non-empty interior.
Then, with the following quantities
\begin{align}
    Q' := QV
    ,\quad
    x' := x - Q\xi
\end{align}
using $\xi$ as in \Cref{lem:polyhedron-low-rank} (and $V$ as in the hypothesis), we have
\begin{align}
    x \in \hull(Q, \sP)
    \iff
    x' \in \hull(Q', \sP')
\end{align}
Note that the conditions of \Cref{lem:caratheodory-hull-full} 
are satisfied trivially for $Q'$ and $\sP'$ by definition.
Moreover, $\dim(\sP') \equiv \dim(\sP)$ by construction.
Applying \Cref{lem:caratheodory-hull-full} for $Q'$, $\sP'$, and $x'$,
we have that $x'$ is a $\sP'$-combination of $Q'$
with an element $w' \in \sP'$ that binds at least 
$\min(\rank(A_{-S}V), (\dim(\sP) - \rank(Q'))_+)$ constraints of $A'$
that form a full row rank sub-matrix of $A'$.
Then, letting $w := \xi + V w'$ and using \labelcref{lem:caratheodory-hull-full:eq:id},
$w$ binds at least $\min(\rank(A_{-S}V), (\dim(\sP)-\rank(Q'))_+)$
constraints of $A_{-S}$, which we denote as $T \subseteq [m-\abs{S}]$.
It remains to show that $\rank(A_{-S}V) = \rank(A) - \rank(A_S)$
and that the binding constraints form a full row rank sub-matrix of $A_{-S}$.
We first show the first claim.
Let $V^\perp \in \R^{p \times \rank(A_S)}$ denote a set of orthogonal 
vectors that span the row space of $A_S$.
Then, $A [V^\perp \quad V]$ has the same rank as $A$
and it is clear that $\rank(AV^\perp) = \rank(A_S)$. 
Moreover, since $A_S V = \zeros$, 
we have $\rank(AV) = \rank(A_{-S} V)$.
Finally, every column of $AV^\perp$ is not in $\cspan(AV)$.
Otherwise, $A_S V^\perp = A_S V M = \zeros$ for some matrix $M$,
which is a contradiction since every column of $V^\perp$ lies in $\cspan(A_S)$.
This proves that $\rank(A_{-S}V) = \rank(A) - \rank(A_S)$.
Next, we show that $(A_{-S})_T$ is a full row rank sub-matrix.
Otherwise, there exists $\gamma \neq \zeros \in \R^{\abs{T}}$ such that
\begin{align}
    (A_{-S})_T^\top \gamma
    =
    \zeros
    \implies
    \zeros 
    = 
    V^\top (A_{-S})_T^\top \gamma
    =
    (A'_T)^\top \gamma
    \implies
    \gamma = \zeros
\end{align}
which is a contradiction.
This completes the proof.
\end{proof}

\subsection{Proof of \Cref{thm:caratheodory-hull-local-uniqueness}}%
\label{app:proofs-main:caratheodory-hull-local-uniqueness}

\begin{proof}
We may assume without loss of generality $\sP$ has a non-empty interior
since otherwise we may first restrict to the lower-dimensional space
removing the maximal affine constraints by \Cref{lem:polyhedron-low-rank}.
Hence, we assume $p \equiv \dim(\sP)$.
We may further assume $Q$ is full row rank so that $n \equiv \rank(Q)$.
We are now in the setting of \Cref{lem:caratheodory-hull-full} where 
$Q$ is full row rank and $\sP$ has a non-empty interior.

Consider the proof of \Cref{lem:caratheodory-hull-full}.
Recall that the proof procedure is to repeatedly find directions $v \in \nullspace(Q)$
such that $w = y + tv$ binds a constraint (where $x = Qy$), say $a_1$.
Let $V \in \R^{p \times (p-1)}$ be as in \Cref{lem:polyhedron-low-rank}
that reduces the polyhedron by one dimension, namely, restricting to the affine space 
$\set{w \in \R^p : a_1^\top w = b_1}$.
We claim that $\rank(QV) = \rank(Q)$.
Indeed, it is equivalent to show that $\dim(\nullspace(QV)) = \dim(\nullspace(Q)) - 1$.
This is easily seen by the fact that 
since $a_1^\top v \neq 0$ by \Cref{lem:polyhedron-intersect} and $V^\top a_1 = \zeros$,
we must have $v \notin \cspan(V)$.
Since the rank is preserved as we apply induction,
we see that eventually $Q$ will be reduced to an invertible matrix,
in which case, the representation for $x$ is unique.
\end{proof}

\end{appendix}

\end{document}